\documentclass [10pt,compsocconf]{IEEEtran}
\usepackage{amsbsy,amsmath,amssymb,epsfig,bbm,mathrsfs,multirow,amsthm}
\usepackage{algpseudocode}
\usepackage{subcaption,float}
\usepackage[hidelinks]{hyperref}

\usepackage[linesnumbered,lined,boxed,ruled,commentsnumbered]{algorithm2e}

\newtheorem{theorem}{\textbf{\textsc{Theorem}}}

\begin{document}
\title{Offloading Energy Efficiency with Delay Constraint for Cooperative Mobile Edge Computing Networks}
\author{\IEEEauthorblockN{Thai T. Vu, Nguyen Van Huynh, Dinh Thai Hoang, Diep N. Nguyen, and Eryk Dutkiewicz\\}
School of Electrical and Data Engineering, University of Technology Sydney, Australia 	\vspace{-5mm}}	

\maketitle

\thispagestyle{empty}
\begin{abstract}
We propose a novel edge computing network architecture that enables edge nodes to cooperate in sharing computing and radio resources to minimize the total energy consumption of mobile users while meeting their delay requirements. To find the optimal task offloading decisions for mobile users, we first formulate the joint task offloading and resource allocation optimization problem as a mixed integer non-linear programming (MINLP). The optimization involves both binary (offloading decisions) and real variables (resource allocations), making it an NP-hard and computational intractable problem. To circumvent, we relax the binary decision variables to transform the MINLP to a relaxed optimization problem with real variables. After proving that the relaxed problem is a convex one, we propose two solutions namely ROP and IBBA. ROP is adopted from the interior point method and IBBA is developed from the branch and bound algorithm. Through the numerical results, we show that our proposed approaches allow minimizing the total energy consumption and meet all delay requirements for mobile users.
\end{abstract}

{\it Keywords-} Task offloading, mobile edge computing, resource allocation, MINLP, and branch-and-bound algorithm. 
 
\section{Introduction}
\label{sec:Introduction}

The development of mobile applications and Internet-of-Things (IoT) networks has brought a great deal of benefits for human lives, but it also faces many challenges. In particular, mobile and IoT applications have been developed recently often require computations with high complexity, e.g., 3D rendering and image processing, and/or low delay constraints, e.g., interactive games and online object recognitions. However, mobile and IoT devices are usually limited by computing resources, battery life, and network connections, and thus advanced applications may not be able to implement on these devices in practice. Thus, mobile edge computing has been introduced as an effective solution to address this problem. 

Mobile Edge Computing (MEC) is an emerging network architecture that ``move" the cloud computing capabilities closer to the mobile users~\cite{Mach2017Mobile}. Specifically, in an MEC network, powerful computing devices, e.g., servers, are deployed at the edges of the mobile network to support hardware resource-constrained devices, e.g., mobile and IoT devices, to perform high complexity computational tasks. The deployment of MEC networks can save energy consumption, increase operation time, and reduce performance delay for smart devices through utilizing powerful resources of the edge nodes. Furthermore, this can reduce operation costs for mobile network operators up to 67\% by reducing the total throughput and peak backhaul bandwidth consumption~\cite{Wang2017Asurvey}. As a result, technical standards for MEC are being developed by the European Telecommunications Standards Institute to promote the development of MEC in future mobile networks~\cite{ETSIWhitePaper2015}. 

However, an MEC node does not possess abundant computing resource as that of the public cloud, e.g., Amazon Web Services and Microsoft Azure. Additionally, although computation offloading demand from mobile users is usually high, not all computational tasks benefit by being offloaded to the edge node. Some tasks even consume more energy when being offloaded than processed locally due to the communication overhead, i.e., transmit requests and receive results. Consequently, joint task offloading and resource allocation to minimize energy consumption for mobile devices under the edge's resource constraints and delay requirements is the most important challenge in MEC networks~\cite{Mach2017Mobile}. 

In~\cite{Chang2017Energy}, the authors study an energy efficient computation offloading scheme in a multi-user MEC system. In particular, the authors first formulate an energy consumption optimization problem with explicit consideration of delay performance. Through analyzing the relationship between mobile users' demands and edge computing node's capacity, the authors then can derive the optimal offloading probability and transmit power for mobile users. Aiming to minimize the overall cost of energy, computation, and delay for all users, the authors in~\cite{Chen2017Joint} introduce a joint offloading and resource allocation for computation and communication in an MEC network. Due to the NP-hard problem, the authors proposed a three-step algorithm including semidefinite relaxation, alternating optimization, and sequential turning. In addition, there are some other research works in the literature studying different approaches for jointly energy efficiency and delay management in MEC networks. For example, the authors in~\cite{Chen2016Efficient} present a computation offloading game model to address the distributed computation offloading decision problem for mobile users, and the authors in~\cite{Cardellini2016game} introduce a computation offloading hierarchical model in which a task can be sent to an MEC node or a cloud server.

Different from all aforementioned work, in this paper, we study a cooperative MEC network in which edge nodes are deployed in the same area to support high complexity computation tasks of the mobile users. The edge nodes have different radio and computing resources, meanwhile mobile users have distinct computation tasks with various delay requirements. To minimize the total energy consumption for mobile users in the network and meet all tasks' delay requirements, we first formulate the joint task offloading and resource allocation optimization problem for all mobile users and edge nodes. Since the optimization problem is a mixed integral non-linear programming (MINLP) which is NP-hard and intractable to solve, we introduce a relaxing solution which converts binary decision variables to real values. We then prove that the relaxed optimization problem is a convex problem which can be solved by some effective methods, e.g., the interior point method (IPM). Although the IPM can find the optimal solution for relaxed problem, the obtained decision variables are real numbers which may not be practical in implementing in the MEC network. In addition, when converting decision variables to real values, the complexity of optimization problem becomes higher, which is inefficiency to implement in MEC networks, especially when the number of variables is large. Therefore, in this paper, we introduce IBBA, an improvement of branch and bound algorithm to address the MINLP. The proposed IBBA allows not only finding optimal binary variables for offloading decisions, but also utilizing the characteristics of binary variables to reduce the complexity in finding the optimal solution. The extensive numerical results are then performed to demonstrate the efficiency of proposed solutions in minimizing the total energy consumption for mobile devices and meeting delay requirements for offloading tasks.

\section{System Model}
\label{sec:sysmodel}

\subsection{Network Model}

We consider an MEC with $N$ mobile users, $M$ cooperative edge nodes, and one cloud server as shown in Fig.~\ref{fig:System-Model}. The set of mobile users and MEC nodes in the network are denoted by $\mathbb{N}=\{1,2,\ldots,N\}$ and $\mathbb{M}=\{1,2,\ldots,M\}$, respectively. Each mobile user has computing tasks which can be processed locally or offloaded to MEC nodes to execute. If a task is decided to be executed at an MEC node, the mobile user will send the requested task to the target edge node. After the task is performed at the edge node, the result will be sent back to the user. Note that if the edge node does not have sufficient computing resources or it cannot meet the delay constraint of the task, the edge node will send the task to the cloud server for processing. 

\begin{figure}[htb!]
	\centering
	\includegraphics[scale=0.55]{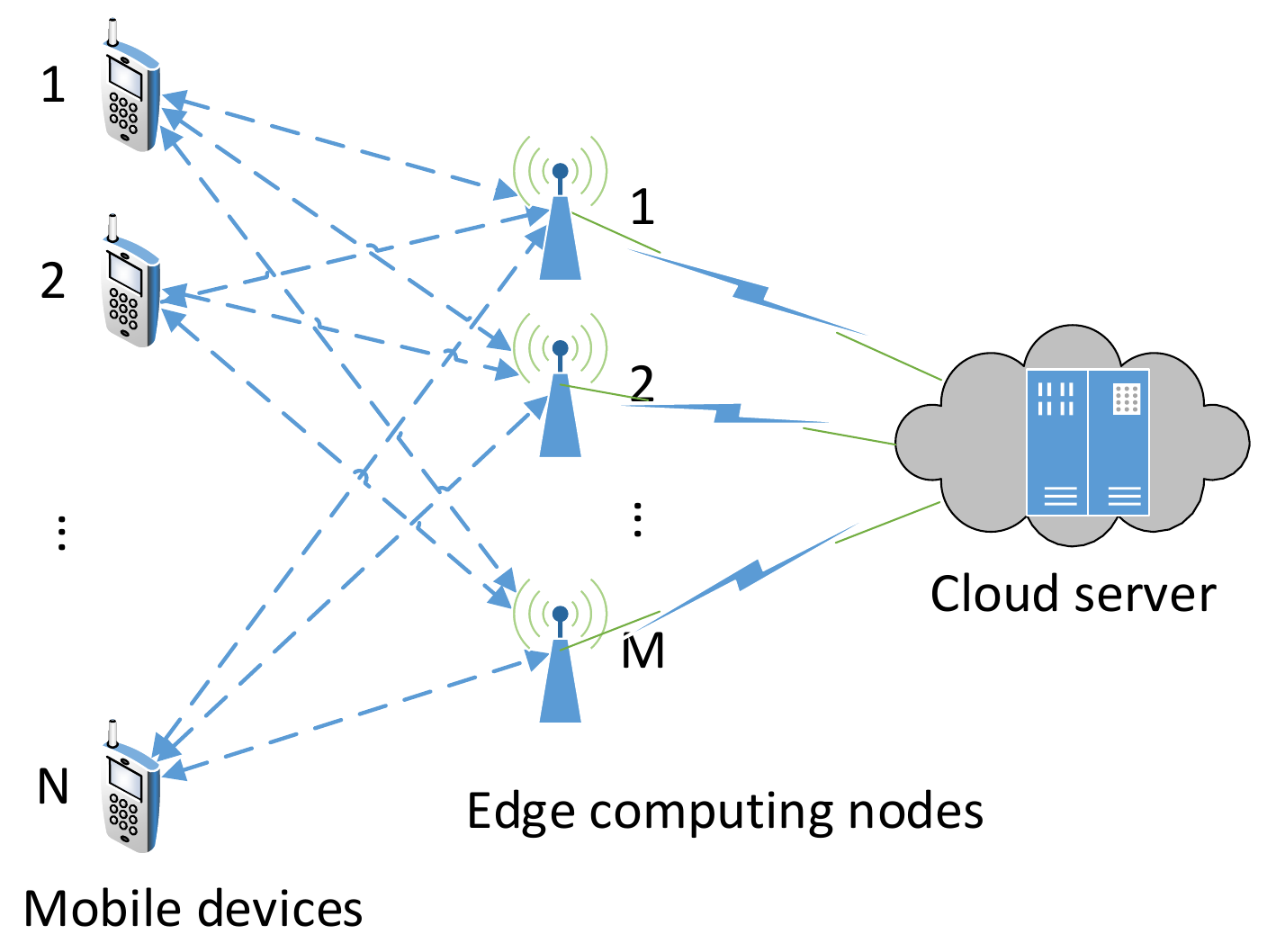}
	\caption{Cooperative mobile edge computing network.}
	\label{fig:System-Model}
\end{figure}

\subsection{Mobile Devices}

At each time slot, mobile user $i$ has a task which needs to be executed. The task is defined by a tuple $I_{i}\left(D_{i}^{i},D_{i}^{o},C_{i},T_{i}^{r}\right)$, in which $D_{i}^{i}$ is the input data size (including input data and execution code), $D_{i}^{o}$ is the output data size, $C_{i}$ is the number of CPU cycles that is required to execute the task, and $T_{i}^{r}$ is the delay requirement of task $I_{i}$. In this paper, we set $T_{i}^{r}$ as the maximum delay requirement of the task. Each mobile device has a processing rate defined by $f_{i}^{l}$ which expresses the hardware capability of the mobile device.

\subsection{MEC Nodes}

We assume that each MEC node $j$ has a resource capability denoted by a tuple $(R_{j}^{u},R_{j}^{d},F_{j}^{f})$ in which $R_{j}^{u}$, $R_{j}^{d}$ and $F_{j}^{f}$ are total uplink rate, total downlink rate, and CPU cycle rate, respectively. These resources can be allocated partially to perform mobile users' offloading tasks.

\subsection{Cloud Server}

All MEC nodes are assumed to be able to connect with a public cloud server. If a task is sent to an MEC node, but the MEC node cannot perform due to the resource or delay constraint, the edge node will forward the task to the cloud server for processing. We denote the data rate between an MEC node and the cloud server as $r^{fc}$, and the the processing rate assigned to each task on the cloud server as $f^{c}$.

\section{Problem Formulation}
\label{sec:LA}

In this paper, we consider a joint offloading and resource allocation problem in which the total energy consumption of mobile devices is minimized. We denote the computation offloading decision variable for task $I_{i}$ by $x_{i}=(x_{i}^{l},x_{i}^{f},x_{i}^{c})$, in which $x_{i}^{l}$, $x_{i}^{f}$ and $x_{i}^{c}$ respectively indicates whether task $I_{i}$ is processed locally at the mobile device, an MEC node, or the cloud server. Here, the variable $x_{i}^{f} = (x_{i1}^{f},x_{i2}^{f},...,x_{iM}^{f})$ is to determine which MEC node will execute the task $I_{i}$. Similarly, the variable $x_{i}^{c} = (x_{i1}^{c},x_{i2}^{c},...,x_{iM}^{c})$ is to determine which MEC node will forward the task to the cloud server.

\subsection{Local Processing}

For the local computing approach, the offloading decision for task $I_{i}$ is defined by $x_{i}=(1,0,0)$. In this case, the consumed energy $E_{i}^{l}$ of the mobile device is proportional to the CPU cycles required for task $I_{i}$ and the expected delay $T_{i}^{l}$ is the execution time of the task. We have:
\begin{equation}
E_{i}^{l}=v_{i}C_{i} ,\quad \text{and} \quad T_{i}^{l}=\frac{C_{i}}{f_{i}^{l}},
\end{equation}
where $v_{i}$ denotes the consumed energy per CPU cycle~\cite{Wen2012Energy}.

\subsection{MEC Node Processing}

For the MEC node processing approach, the offloading decision for task $I_{i}$ is defined by $x_{i}=(0,1,0)$. If task $I_{i}$ is processed at MEC node $j$ ($x_{ij}^{f}=1$), the MEC node will allocate spectrum and computation resources for mobile device $i$, defined by a tuple $r_{ij}=(r_{ij}^{u},r_{ij}^{d},f_{ij}^{f})$, in which $r_{ij}^{u}$, $r_{ij}^{d}$ respectively are uplink rate, downlink rate for input and output transmissions, and $f_{ij}^{f}$ is CPU cycle rate for the task being processed at MEC node $j$. In this case, the energy consumption at the mobile user is for both transferring data to and receiving data from the MEC node $j$, and the delay includes time for transmitting input data, receiving output data and task processing at the MEC node.

Let $e_{ij}^{u}$ and $e_{ij}^{d}$ denote the energy consumption for transmitting and receiving a unit of data, respectively. The consumed energy of mobile device $E_{ij}^{f}$ and the delay $T_{ij}^{f}$ are given by:
\begin{equation}
E_{ij}^{f}=E_{ij}^{u}+E_{ij}^{d}, \quad \text{and} \quad T_{ij}^{f}=\frac{D_{i}^{i}}{r_{ij}^{u}}+\frac{D_{i}^{o}}{r_{ij}^{d}}+\frac{C_{i}}{f_{ij}^{f}} ,
\end{equation}
where $E_{ij}^{u}=e_{ij}^{u}D_{i}^{i}$ and $E_{ij}^{d}=e_{ij}^{d}D_{i}^{o}.$

Additionally, task $I_{i}$ can be processed at only one MEC node, and thus the consumed energy $E_{i}^{f}$ of mobile device and the delay $T_{i}^{f}$ since task $I_{i}$ is processed at MEC node $j$ is defined as follows:
\begin{equation}
\begin{aligned}
E_{i}^{f}=\sum_{j=1}^{M}x_{ij}^{f}E_{ij}^{f} , \quad \text{and} \quad T_{i}^{f}=\sum_{j=1}^{M}x_{ij}^{f}T_{ij}^{f},
\end{aligned}
\end{equation}
s.t.
\begin{equation}
\begin{aligned}
\left\{	\begin{array}{ll}
x_{i}^{f}=\sum_{j=1}^{M}x_{ij}^{f} = 1,	\\
x_{ij}^{f} \in\{0,1\}, \forall j\in \mathbb{M}.	\\
\end{array}	\right.
\end{aligned}
\label{eq:offloading_at_fog_node_constraints}
\end{equation}

\subsection{Cloud Server Processing}

For the cloud computing approach, the offloading decision for task $I_{i}$ is defined by $x_{i}=(0,0,1)$. If MEC node $j$ forwards task $I_{i}$ to the cloud server (i.e., $x_{ij}^{c}=1$), the MEC node will allocate communication resource for mobile device~$i$, defined by a tuple $r_{ij}=(r_{ij}^{u},r_{ij}^{d},f_{ij}^{f})$, in which $r_{ij}^{u}$, $r_{ij}^{d}$ are uplink rate, downlink rate for input and output transmissions, and $f_{ij}^{f}=0$. After receiving the task, the MEC node $j$ sends the input data to the cloud server for processing, then receives and sends the result back to the mobile device. In this case, the total consumed energy $E_{ij}^{c}$ at the mobile user is the same as in the case of the MEC processing, while the delay $T_{ij}^{c}$ includes the time for transmitting the input from mobile user to the MEC node, time from the MEC node to the cloud server, time for receiving the output from the cloud server to mobile user via the edge node, and task-execution time at the cloud server. These performance metrics are as follows:
\begin{equation}
E_{ij}^{c}=E_{ij}^{f}=E_{ij}^{u}+E_{ij}^{d} ,
\label{eq:cloud_energy}
\end{equation}

and
\begin{equation}
T_{ij}^{c}=\frac{D_{i}^{i}}{r_{ij}^{u}}+\frac{D_{i}^{o}}{r_{ij}^{d}}+
\frac{(D_{i}^{i}+D_{i}^{o})}{r^{fc}}+\frac{C_{i}}{f^{c}} .
\label{eq:cloud_delay}
\end{equation}

Similarly, because only one MEC node can forward task $I_{i}$ to the cloud server, the consumed energy $E_{i}^{c}$ of mobile device and the delay $T_{i}^{c}$ since task $I_{i}$ is processed at the MEC node are defined as follows:
\begin{equation}
\begin{aligned}
E_{i}^{c}=\sum_{j=1}^{M}x_{ij}^{c}E_{ij}^{c}, \quad \text{and} \quad T_{i}^{c}=\sum_{j=1}^{M}x_{ij}^{c}T_{ij}^{c},
\end{aligned}
\end{equation}
s.t.
\begin{equation}
\begin{aligned}
\left\{	\begin{array}{ll}
x_{i}^{c}=\sum_{j=1}^{M}x_{ij}^{c} = 1,	\\
x_{ij}^{c} \in\{0,1\}, \forall j\in \mathbb{M}.	\\
\end{array}	\right.
\end{aligned}
\label{eq:offloading_at_cloud_constraints}
\end{equation}

Let $E_{i}$ and $T_{i}$, respectively, be the consumed energy of mobile device and the delay when task $I_{i}$ is processed. Note that a task can be executed at either the mobile device, an MEC node, or the cloud server. Thus, we have:
\begin{eqnarray}
E_{i}=x_{i}^{l}E_{i}^{l}+x_{i}^{f}E_{i}^{f} +x_{i}^{c}E_{i}^{c} , \label{eq:energy} \\
T_{i}=x_{i}^{l}T_{i}^{l}+x_{i}^{f}T_{i}^{f} +x_{i}^{c}T_{i}^{c}, \label{eq:delay}
\end{eqnarray}
s.t.
\begin{equation}
\begin{aligned}
\left\{	\begin{array}{ll}
x_{i}^{l}+x_{i}^{f}+x_{i}^{c}=1,	\\
x_{i}^{l}, x_{i}^{f}, x_{i}^{c} \in\{0,1\}.
\end{array}	\right.
\end{aligned}
\label{eq:offloading_constraints}
\end{equation}

From (\ref{eq:offloading_at_fog_node_constraints}), (\ref{eq:offloading_at_cloud_constraints}) and (\ref{eq:offloading_constraints}), we derive the following constraints:
\begin{equation}
\begin{aligned}
\left\{	\begin{array}{ll}
x_{i}^{l}+x_{i}^{f}+x_{i}^{c}= x_{i}^{l}+\sum_{j=1}^{M}x_{ij}^{f}+\sum_{j=1}^{M}x_{ij}^{c}=1,	\\
x_{i}^{l}, x_{i}^{f}, x_{i}^{c} \in\{0,1\},	\\
x_{ij}^{f}, x_{ij}^{c} \in\{0,1\}, \forall (i,j)\in \mathbb{N}\times\mathbb{M}.
\end{array}	\right.
\end{aligned}
\label{eq:all_offloading_constraints}
\end{equation}

%

In this paper, we address the joint offloading decision $\{x_{i}\}$ and resource allocation $\{r_{i}\}$ problem in which the objective is to minimize the total energy consumption of all mobile devices and all delay constraints must be satisfied, i.e.,
\begin{equation}
(\mathbf{P}_1) \phantom{10}	 \underset{\{x_{i}\},\{r_{i}\}}{min}\sum_{i=1}^{N}E_{i} ,
\label{eq:global_energy}
\end{equation}
s.t.
\begin{equation}
\begin{aligned}
\left\{	\begin{array}{ll}
x_{i}^{l}+\sum_{j=1}^{M}x_{ij}^{f}+\sum_{j=1}^{M}x_{ij}^{c}=1,	\\
x_{i}^{l},x_{ij}^{f}, x_{ij}^{c} \in\{0,1\}, \forall (i,j)\in \mathbb{N}\times\mathbb{M},		\\
T_{i} \leq T_{i}^{r}, \forall i\in \mathbb{N},	\\
\sum_{i=1}^{N}f_{ij}^{f} \leq F_{j}^{f}, \forall j\in \mathbb{M},	\\
\sum_{i=1}^{N}r_{ij}^{u} \leq R_{j}^{u}, \forall j\in \mathbb{M},	\\
\sum_{i=1}^{N}r_{ij}^{d} \leq R_{j}^{d}, \forall j\in \mathbb{M},	\\
r_{ij}^{u}, r_{ij}^{d}, r_{ij}^{f} \geq 0,\forall (i,j)\in \mathbb{N}\times\mathbb{M}.
\end{array}	\right.
\end{aligned}
\label{eq:resource_and_delay_constraints}
\end{equation}

The optimization problem $(\mathbf{P}_1)$ is an NP-hard. Hence, standard optimization techniques cannot be applied directly and the globally optimal solution is unfeasible. Thus, we introduce two effective approaches to address this problem.

\section{Proposed Optimal Solutions}

\subsection {Relaxing Optimization Solution}
In this section, we introduce a relaxing approach which allows to find the optimal solution through converting binary decision variables, i.e., $\{x_{i}\}$, to real variables. By relaxing binary variables to real numbers, we then can reformulate the optimization problem $(\mathbf{P}_1) $ as follows:
\begin{equation}
(\mathbf{P}_2) \phantom{10} \underset{\{x_{i}\},\{r_{i}\}}{min}\sum_{i=1}^{N}E_{i} ,
\label{eq:global_energy_relax}
\end{equation}
s.t.
\begin{equation}
\begin{aligned}
\left\{	\begin{array}{ll}
x_{i}^{l}+\sum_{j=1}^{M}x_{ij}^{f}+\sum_{j=1}^{M}x_{ij}^{c}=1,	\\
x_{i}^{l},x_{ij}^{f}, x_{ij}^{c} \in[0,1], \forall (i,j)\in \mathbb{N}\times\mathbb{M},		\\
T_{i} \leq T_{i}^{r}, \forall i\in \mathbb{N},	\\
\sum_{i=1}^{N}f_{ij}^{f} \leq F_{j}^{f}, \forall j\in \mathbb{M},	\\
\sum_{i=1}^{N}r_{ij}^{u} \leq R_{j}^{u}, \forall j\in \mathbb{M},	\\
\sum_{i=1}^{N}r_{ij}^{d} \leq R_{j}^{d}, \forall j\in \mathbb{M},	\\
r_{ij}^{u}, r_{ij}^{d}, r_{ij}^{f} \geq 0,\forall (i,j)\in \mathbb{N}\times\mathbb{M},
\end{array}	\right.
\end{aligned}
\label{eq:resource_and_delay_constraints_relax}
\end{equation}

To find the optimal solution for $(\mathbf{P}_2)$, we will prove that the relaxed problem is a convex optimization problem.

\begin{theorem}
\label{theo:convexity}
The relaxed problem $(\mathbf{P}_2)$ is a convex optimization problem.
\end{theorem}


\begin{proof} From (\ref{eq:energy}), the energy consumption of task $i$, $E_i$, is a linear function of decision variable $x_i$. Consequently, the objective function $\sum_{1}^{N}E_{i}$ is a linear function with respect to real decision variables $\{x_{i}\}$. From (\ref{eq:delay}), the delay $T_{i}$ is the sum of linear and linear-fractional functions: $x_{i}^{l}$, $\frac{x_{ij}^{f}}{r_{ij}^{u}}$, $\frac{x_{ij}^{f}}{r_{ij}^{d}}$, $\frac{x_{ij}^{f}}{f_{ij}^{f}}$, $x_{ij}^{c}$, $\frac{x_{ij}^{c}}{r_{ij}^{u}}$ and $\frac{x_{ij}^{c}}{r_{ij}^{d}}$ for all $j$ in $\mathbb{M}$, These functions have positive coefficients: $C_i$, $D_i^i$, $D_i^o$, $\left(\frac{D_i^i+D_i^o}{f^{fc}}+\frac{C_i}{f^c}\right)$, $D_i^i$ and $D_i^o$, respectively. Thus, $T_{i}(x_i, r_i)$ is a concave function with respect to $x_i$ and $r_i$~\cite{Boyd2004Convex}. Since the objective function in (\ref{eq:global_energy_relax}) is a linear function, and the constraints in (\ref{eq:resource_and_delay_constraints_relax}) are concave functions, the relaxed problem is a convex optimization problem~\cite{Boyd2004Convex}.
\end{proof}

To solve the relaxed optimization problem $(\mathbf{P}_2)$, we can apply some effective tools as mentioned in~\cite{Boyd2004Convex}. In this paper, we adopt the \emph{interior-point method}~\cite{Boyd2004Convex} to find the optimal solution because this is a very effective tool to address the convex optimization problem with constraints.

\subsection{Improved Branch and Bound Algorithm}

Although the relaxing approach can address the join offloading and resource allocation $(\mathbf{P}_1)$, its obtained optimal decision variables are real numbers which are impractical to implement in MEC networks. Furthermore, the relaxing approach cannot utilize the advantage of binary variables in reducing the complexity and finding the optimal solution. In particular, binary variables have only two variables, i.e., either $0$ or $1$. In addition, when the value of a variable is zero, its product will be zero, which allows to reduce the computational complexity significantly. Thus, we introduce an improved branch and bound algorithm, namely IBBA, which allows not only addressing the MINLP, but also utilizing the characteristics of binary variables to reduce the complexity of optimization problem $(\mathbf{P}_1)$. 

In this paper, we exploit the following properties of the optimization problem~$(\mathbf{P}_1)$ to propose the IBBA.

\begin{itemize}
\item \label{Branching_Task} \textbf{Branching task} dictates that a task can be executed at only one place, i.e, at the mobile device, one of edge nodes, or the cloud server via an MEC node. Thus, for the offloading decisions $x_i$ there is only one variable that is equal to $1$, and all others are equal to $0$. Thus, at a node in the IBBA tree, we choose to branch the decisions of a task, forming a \textit{(2M+1)}-tree with height $N$. 
\item \label{Simplifying problem} \textbf{Simplifying problem} dictates that when a task is executed at mobile device, an edge node, or the cloud server via an edge node, all other MEC nodes do not need to allocate resources toward that task. Thus, when $x_{ij}^f = 0$ or $x_{ij}^c = 0$, we can eliminate all sub-expressions of the forms $x_{ij}^fA$ and $x_{ij}^cB$, these decision variables, and related resource allocation variables $f_{ij}^f$, $r_{ij}^u$ and $r_{ij}^d$ in~$(\mathbf{P}_1)$. Consequently, we have sub-problems with the reduced number of variables.
\item \label{Preserving Convexity} \textbf{Preserving convexity} dictates that after fixing some binary variables, sub-problems are convex optimization problems. In particular, based on Theorem~\ref{theo:convexity}, it can be observed that if we fix one or multiple binary variables in~$(\mathbf{P}_1)$ and set all other variable to be real variables, the corresponding relaxed sub-problems are always convex. 
\end{itemize}

Based on three aforementioned properties, we introduce Algorithm~\ref{IBBA_Algorithm_code}. This algorithm not only allows to find the optimal solution for the optimization problem~$(\mathbf{P}_1)$ faster, but also provides optimal binary offloading decision variables which can be efficiently implemented in MEC networks in practice. 

\begin{algorithm}
	\DontPrintSemicolon
	\SetKwInput{Left}{left}\SetKwData{This}{this}\SetKwData{Up}{up}
	\SetKwInOut{Input}{Input}\SetKwInOut{Output}{Output}
	\Input{Set of tasks $\{I_{i}\left(D_{i}^{i},D_{i}^{o},C_{i},T_{i}^{r}\right)\}$\\Set of MEC nodes $\{Node_j(R_{j}^{u},R_{j}^{d},F_{j}^{f})\}$\\Cloud server $r^{fc}$, $f^c$}
	\Output{Optimal variables of problem $(\mathbf{P}_1)$ }
	\BlankLine
	\Begin{
		$Solution \gets \emptyset$; $optVal \gets +\infty$\;
		$Stack.empty()$; $Stack.push$($(\mathbf{P}_1)$)\;
		\While{$Stack.isNotEmpty()$}{
			$curProb \gets Stack.pop()$\;
			$tempSol, tempVal \gets$ \underline{Solve} relaxing problem of $curProb$\;
			\If{$tempVal > optVal$ {\bf or} $curProb$ \textnormal{is infeasible}}{
				\underline{Prune} $curProb$\;
			}
			\If{$tempVal < optVal$}{
				\If{$tempSol$ \textnormal{satisfies all integer constraints of} $\{x_i\}$}{
					$Solution \gets tempSol$\;
					$optVal \gets tempVal$\;
					\underline{Prune} $curProb$\;
				}\Else{
					$subProblems \gets$ \underline{Branch} $curProb$ by fixing the decisions of the first task in the set $\{I_i\}$, which is not fixed so far, based on \textbf{Branching task} property.\;
					\For{{\bf{each}} $subProb$ {\bf in} $subProblems$}{
						\underline{Simplify} $subProb$ based on \textbf{Simplifying problem} property.\;
						$Stack.push(subPob)$\;
					}
				}
			}
		}
		\textbf{Return} $Solution$ and $optVal$\;
	}
	\caption{IBBA Algorithm\label{IBBA_Algorithm_code}}
\end{algorithm}

\subsection{Offloading Analysis}
Before conducting experiments, we analyze when mobile users can benefit from offloading. A mobile user is said to be benefit from offloading if its total energy consumption when the task is offloaded is lower than processing locally. When the task is processed at the mobile device, the consumed energy depends on the required CPU cycles for the task. However, if the task is offloaded, the consumed energy at the mobile device is for both transferring input data $D_i^i$ to and receiving output data $D_i^o$ from an MEC node, thus the energy depends only on the input and output data sizes. Thus, for the task $i$, offloading will benefit if $E_i^l > E_i^f$. While $E_i^l$ is a function of required CPU cycles, $E_i^f$ is a function of input/output data sizes. Therefore, we introduce parameter $\alpha$ as ratio between the number of required CPU cycles and input data size in order to quantify the likelihood of offloading tasks. Let $\alpha_i^*$ be the task complexity ratio at which $E_i^l = E_i^f$. We have:
\begin{equation}
\alpha_i^* = \frac{e_{ij}^uD_i^i + e_{ij}^dD_i^o}{v_i D_i^i}.
\label{eq:complexity_rate}
\end{equation}

Let $\alpha_i$ be the ratio between the number of required CPU cycles $C_i$ and input data size $D_i^i$. We have $C_i = \alpha_i \times D_i^i$. Thus, task $i$ is likely to be offloaded if $E_i^l > E_i^f$ or $\alpha_i > \alpha_i^*$. This parameter is especially important in evaluating offloaded tasks as well as analyzing the performance of whole system. 

\section{Performance Evaluation} 

\subsection{Experiment Setup} 
We use the configuration of a Nokia N900 mobile device described in \cite{Miettinen2010Energy}, and set the number of devices as $N = 10$. Each mobile device has CPU rate $f_i^l=0.5$ Giga cycles/s and the unit processing energy consumption $v_i=\frac{1000}{730}$ J/Giga cycle. We denote $U(a,b)$ as discrete uniform distribution between $a$ and $b$. Here, we assume that each device has a task with the input and output data sizes following uniform distributions $U(10,20)$MB and $U(1,2)$MB, respectively. We also assume that each task has required $C_i$ CPU processing cycles defined by $\alpha_i \times D_i^i$ Giga cycles, in which the parameter $\alpha_i$ Giga cycles/MB is the complexity ratio of the task. 
All parameters are given in Table~\ref{tab:Experimental-parameters}.

\begin{table}[htbp]
	\caption{Experimental parameters\label{tab:Experimental-parameters}}
	\centering{}%
	\begin{tabular}{|l|c|}
		\hline 
		\textbf{Parameters} & \textbf{Value}\tabularnewline
		\hline 
		Number of mobile devices $N$ & 10\tabularnewline
		\hline 
		Number of MEC nodes $M$ & 4\tabularnewline
		\hline 
		CPU rate of mobile devices $f_i^l$ & $0.5$ Giga cycles/s\tabularnewline
		\hline
		Processing energy consumption rate $v_i$ & $\frac{1000}{730}$ J/Giga cycles\tabularnewline
		\hline  
		Input data size $D_i^i$ & $U(10,20)$ MB\tabularnewline
		\hline
		Output data size $D_i^o$ & $U(1,2)$ MB\tabularnewline
		\hline
		Required CPU cycles $C_i$ & $\alpha_i \times D_i^i$ \tabularnewline
		\hline
		Unit transmission energy consumption $e_{ij}^u$ & $0.142$  J/Mb\tabularnewline
		\hline
		Unit receiving energy consumption $e_{ij}^d$ & $0.142$  J/Mb\tabularnewline
		\hline 
		Delay requirement $T_i^r$  & $[30,60]$s\tabularnewline
		\hline		
		Processing rate of each MEC node $F_j^f$ & $10$ Giga cycles/s\tabularnewline
		\hline
		Uplink data rate of each MEC node $R_i^u$ & 72 Mbps\tabularnewline
		\hline 
		Downlink data rate of each MEC node $R_i^d$ & 72 Mbps\tabularnewline
		\hline  
		CPU rate of the cloud server $f^c$ & $10$ Giga cycles/s\tabularnewline
		\hline 
		Data rate between FNs and the cloud $r^{fc}$ & $5$ Mbps\tabularnewline
		\hline  
	\end{tabular}
\end{table}
Here, we refer the policy in which all tasks are processed locally as ``Without Offloading" (WOP), and the policy in which all tasks are offloaded to the MEC nodes or the cloud server as the ``All Offloading" (AOP). The results obtained by Algorithm~\ref{IBBA_Algorithm_code} (IBBA) will be compared with the relaxing optimization policy (ROP), WOP, and AOP.

\subsection{Numerical Results}

\subsubsection{Scenario 1 - Vary the Complexity of Tasks} 

In this scenario, we investigate the effect of task complexity on the offloading decisions and energy consumption of mobile devices by varying the complexity of all tasks. At first, we choose the complexity ratio of tasks $\alpha_i$ as $U(200, 500)$ cycles/byte, then increase each task $100$ cycles/byte for each experiment. The delay requirement is set at $40$s for all tasks. Other parameters are set as in Table~\ref{tab:Experimental-parameters}. 

\begin{figure}[htb!]
	\centering
	\includegraphics[scale=0.55]{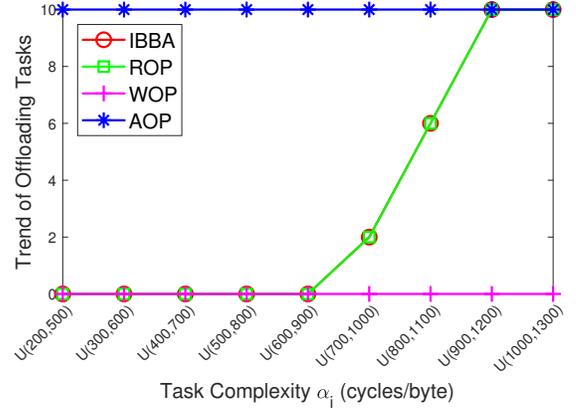}
	\caption{Trend of offloading as the task complexity $\alpha_i$ increases.}
	\label{fig:offloading_vary_complexity}
\end{figure}

Fig.~\ref{fig:offloading_vary_complexity} depicts the trend of offloading tasks when the task complexity ratio $\alpha_i$ is increased. While the trend of offloading tasks of the WOP and AOP are constants, i.e., $0$ and $10$, respectively, the offloading trends of both IBBA and ROP go up as $\alpha_i$ increases. Specifically, the numbers of offloaded tasks of the IBBA and ROP are equal $0$ as the complexity ratio increases from $U(200, 500)$ to $U(600,900)$. This is because $\alpha_i$ is less than $\alpha_i^*$, which is equal to $911$ cycles/byte according to Eq.~(\ref{eq:complexity_rate}) and parameters in Table \ref{tab:Experimental-parameters}. Moreover, all tasks executed locally still can satisfy the delay constraints ($T_i^r = 40$s). Then, the numbers of offloading tasks increase dramatically from $0$ to $10$ since the task complexity ratio $\alpha_i$ increases from $U(600,900)$ to $U(1000,1300)$. This is because there is an increasing number of tasks with $\alpha_i>\alpha_i^*$. Noticeably, Fig.~\ref{fig:offloading_vary_complexity} also shows that all tasks get benefit from offloading when $\alpha_i$ is in the ranges from $U(900,1200)$ to $U(1000,1300)$.

Fig.~\ref{fig:energy_vary_complexity} shows the average energy consumption of mobile devices for IBBA, ROP, WOP and AOP, when $\alpha_i$ increases from $U(200,500)$ to $U(1000,1300)$. Generally, while the average energy consumption is a constant ($18.4$J/task) for the AOP, it increases for other methods. This is because in the AOP, all tasks are offloaded and the consumed energy at mobile devices depends only on the data sizes of $D_i^i$ and $D_i^o$. For the WOP, the consumed energy increases linearly according to the task complexity ratio. Similar to Fig.~\ref{fig:offloading_vary_complexity}, the energy consumption trends of the IBBA and ROP are the same because their offloading decisions are impacted by the energy efficiency factor without constraints.

\begin{figure}[htb!]
	\centering
	\includegraphics[scale=0.55]{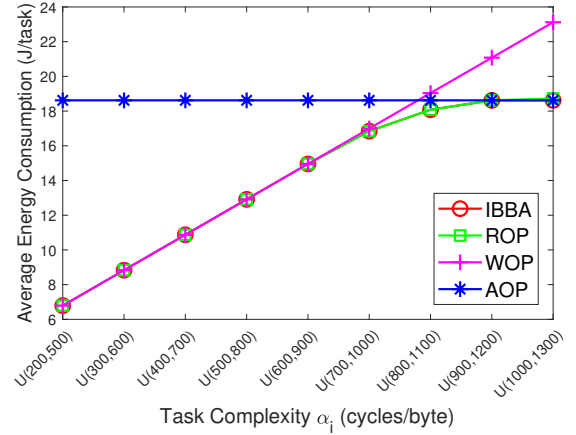}
	\caption{Average energy consumption of mobile devices as the task complexity $\alpha_i$ is increased.}
	\label{fig:energy_vary_complexity}
\end{figure}

\begin{figure}[htb!]
	\centering
	\includegraphics[scale=0.55]{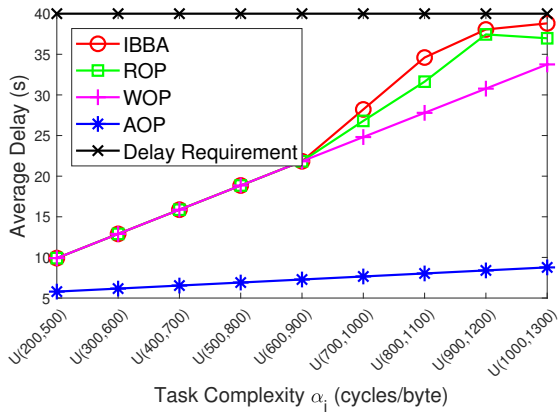}
	\caption{Average task processing delay as the task complexity $\alpha_i$ is increased.}
	\label{fig:delay_vary_complextiy}
\end{figure}

Fig.~\ref{fig:delay_vary_complextiy} shows the average delay as the task complexity is increased. Generally, the average delay increases for all policies. Remarkably, the average delays of the IBBA and ROP are always lower than the delay requirement $T_i^r=40$s. From the average delay of the AOP policy, we can observe that the offloading computation can support all tasks with less than $10$s of the delay requirement $T_i^r$.
Moreover, the AOP policy also gains the lowest average delay due to the low average complexity and data size of the tasks.

\subsubsection{Scenario 2 - Vary the Task Delay Requirements} 
In this scenario, we study the impact of task delay requirements on the energy consumption and offloading decisions of mobile devices. We keep the settings as in Table \ref{tab:Experimental-parameters}, and select a set of the tasks with complexity $\alpha_i$ following $U(800, 1100)$ from Scenario 1. Specifically, there are $6$ tasks receiving benefits from offloading due to $\alpha_i > \alpha_i^*=911$ cycles/byte. We then change input/output data sizes of one task so that even it does not get benefit from offloading, but its local processing delay $T_i^l$ is greater than $60$s. The delay requirement $T_i^r$ for all tasks increases from $30$s to $60$s.

\begin{figure}[htb!]
	\centering
	\includegraphics[scale=0.55]{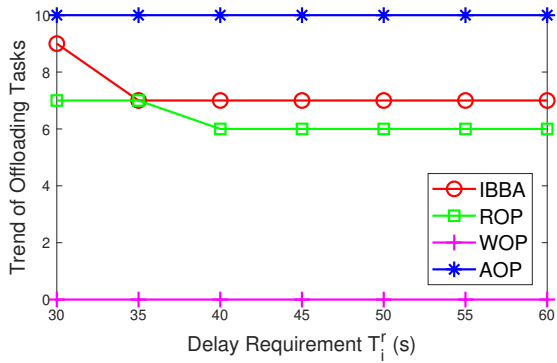}
	\caption{Trend of offloading as the delay requirement is looser.}
	\label{fig:offloading_vary_delay_requirement}
\end{figure}

Figs.~\ref{fig:offloading_vary_delay_requirement}~and~\ref{fig:energy_vary_delay_requirement} illustrate the trend of offloading tasks and average energy consumption, respectively. As observed in Fig.~\ref{fig:offloading_vary_delay_requirement}, at $T_i^r=30$s, while IBBA has $9$ offloaded tasks, ROP has only $7$ offloaded tasks. As mentioned before, there are only $6$ beneficial tasks from offloading, and IBBA algorithm always returns the optimal integer solution. Thus, $9$ offloaded tasks in IBBA including $6$ tasks which get benefits from offloading and $3$ tasks with the local processing delay $T_i^l$ greater than $T_i^r$. Similarly, while the IBBA maintains $7$ offloaded tasks including $6$ beneficial tasks and a task with local delay $T_i^l>60$s as $T_i^r$ increases from $40$s to $60$s, ROP offloads only $6$ beneficial tasks. Consequently, in Fig.~\ref{fig:energy_vary_delay_requirement}, the consumed energy of ROP is always lower than IBBA, the actual MINLP solution. The ROP has to pay for this by having a proportion of tasks that will not satisfy the constraints. In summary, in both IBBA and ROP, when the delay requirements are looser, tasks without benefit from offloading, tend to be processed locally aiming at reducing the consumed energy.

\begin{figure}[htb!]
	\centering
	\includegraphics[scale=0.55]{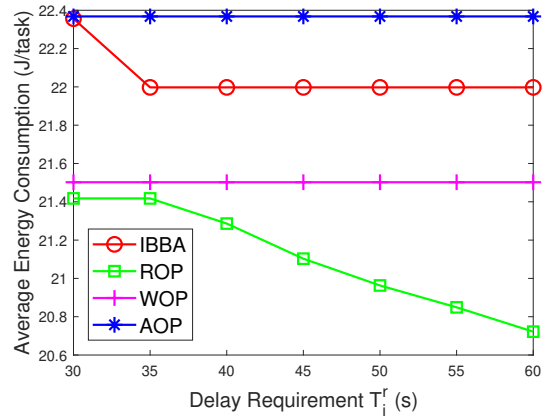}
	\caption{Average consumed energy at mobile devices when the delay requirement is looser.}
	\label{fig:energy_vary_delay_requirement}
\end{figure}

\section{Summary} 

In this paper, we study the offloading problem for the cooperative mobile edge computing network in which mobile edge nodes cooperate to perform computation requirements of the mobile users. To minimize the total energy consumption and meet all delay requirements of mobile users, we formulate the joint offloading decision and resource allocation optimization problem, and propose two effective methods, i.e., IBBA based on the Branch and bound method and ROP based on the interior point method, to find the optimal solution for both the mobile users and edge nodes. The numerical results then verify the efficiency of the proposed solutions. 

\bibliographystyle{IEEE}

\begin{thebibliography}{1}

\bibitem{Mach2017Mobile}
P.~Mach and Z.~Becvar, ``Mobile edge computing: A survey on architecture and computation offloading,'' \emph{IEEE COMST}, vol. 19, no. 3, pp. 1628-1656, Mar. 2017.

\bibitem{Wang2017Asurvey}
S.~Wang, X.~Zhang, Y.~Zhang, L.~Wang, J.~Yang, and W.~Wang, ``A survey on mobile edge networks: Convergence of computing, caching and communications,'' \emph{IEEE Access}, vol. 5, pp. 6757-6779, Mar. 2017.

\bibitem{ETSIWhitePaper2015}
ETSI White Paper No.~11, Mobile Edge Computing: A Key Technology Towards 5G. Available Online: \url{http://www.etsi.org/images/files/ETSIWhitePapers/etsi_wp11_mec_a_key_technology_towards_5g.pdf}

\bibitem{Chang2017Energy}
Z.~Chang, Z.~Zhou, T.~Ristaniemi, and Z.~Niu, ``Energy efficient optimization for computation offloading in fog computing system,'' in \emph{IEEE GLOBECOM}, pp. 1–6, Singapore, Dec. 2017.

\bibitem{Chen2017Joint}
M.~H.~Chen, B. Liang, and M. Dong, ``Joint offloading and resource allocation for computation and communication in mobile cloud with computing access point,'' in \emph{IEEE ICC}, pp. 1–9, Atlanta, USA, May 2017.

\bibitem{Chen2016Efficient}
X.~Chen, L.~Jiao, W.~Li, and X.~Fu, ``Efficient multi-user computation offloading for mobile-edge cloud computing,'' \emph{IEEE/ACM ToN}, vol. 24, no. 5, pp. 2795–2808, Oct. 2016.

\bibitem{Cardellini2016game}
V.~Cardellini, V.~De~N.~Person, V.~Di~Valerio, F.~Facchinei, V.~Grassi, F.~Lo~Presti, and V.~Piccialli, ``A game-theoretic approach to computation offloading in mobile cloud computing,'' \emph{Mathematical Programming}, vol. 157, no. 2, pp. 421–449, Jun. 2016. 

\bibitem{Wen2012Energy}
Y.~Wen, W.~Zhang, and H.~Luo, ``Energy-optimal mobile application execution: Taming resource-poor mobile devices with cloud clones,'' in \emph{IEEE INFOCOM}, pp. 2716–2720, 2012.

\bibitem{Boyd2004Convex}
S.~Boyd and L.~Vandenberghe, \emph{Convex Optimization}. Cambridge university press, 2004.

\bibitem{Miettinen2010Energy}
A.~P.~Miettinen and J.~K.~Nurminen, ``Energy efficiency of mobile clients in cloud computing,'' \emph{HotCloud}, vol. 10, pp. 1–4, 2010.



\end{thebibliography}

\end{document}